\documentclass[12pt]{article}
\usepackage[margin=1in]{geometry}

\usepackage{style/qstyle}

\DeclareMathOperator{\PME}{PME}
\DeclareMathOperator{\FD}{FD}
\DeclareMathOperator{\NPA}{NPA}

\title{On the undecidability of quantum channel capacities}

\author[1]{Archishna Bhattacharyya\footnote{abhat086@uottawa.ca}}
\author[1]{Arthur Mehta\footnote{amehta2@uottawa.ca}}
\author[2]{Yuming Zhao\footnote{yuming@math.ku.dk}}
\affil[1]{University of Ottawa, Department of Mathematics and Statistics}
\affil[2]{QMATH, Department of Mathematical Sciences, University of Copenhagen}

\date{\today}

\begin{document}

\maketitle

\vspace{0.5cm}

\begin{abstract}
    An important distinction in our understanding of capacities of classical versus quantum channels is marked by the following question: \emph{is there an algorithm which can compute (or even efficiently compute) the capacity?} While there is overwhelming evidence suggesting that quantum channel capacities may be uncomputable, a formal proof of any such statement is elusive. We initiate the study of the hardness of computing quantum channel capacities. We show that, for a general quantum channel, it is \textsf{QMA}-hard to compute its quantum capacity, and that the entanglement-assisted zero-error capacity under some restrictions is uncomputable; indicative of the fact that quantum channel capacities may generally be undecidable.
\end{abstract}

\vspace{1cm}

\tableofcontents

\section{Introduction} \label{sec: intro}

Quantum channels represent dynamical processes occurring in Nature from an information-theoretic viewpoint. Quantifying the rate of transmission through a channel is equivalent to ascertaining the amount of resources consumed or created in a dynamical process and underpins much of information theory founded from the groundbreaking work of Shannon \cite{Sha48}. The maximum rate of transmission is called the capacity of the channel. Informally, the capacity is the ratio of the number of bits over the number of channel uses, where a bit measures the amount of entropy or information. While for classical channels there exists only the classical capacity, for quantum channels the landscape becomes richer, as one may send classical or quantum bits with or without the assistance of resources such as entanglement, symmetric side channels, and others. Moreover, the data sent may or may not be private, i.e., shielded from an adversary and only accessible to the receiver. Depending on these properties, the capacities may be classical, private, or quantum with or without assistance. 

On the basis of whether information transmission through a channel is perfect or imperfect one can further restrict the fidelity of the state at the channel input with that at the output to be unity. Such a model depicts perfect transmission characterising \textit{zero-error} capacities. Interestingly enough, it was again Shannon's work \cite{Sha56} founding the notion of zero-error capacities by associating to a channel its so-called \textit{confusability graph}, encoding many of its properties. In the quantum case, Duan-Severini-Winter \cite{DSW12} introduced not only the notion of zero-error quantum capacities, but presented them via a noncommutative generalisation of a graph, associating to every channel $\Phi$ its (quantum) confusability graph or \textit{operator system} $S_{\phi}$. 

In terms of our understanding of classical versus quantum channel capacities an important distinction is marked by the following question:
\begin{question}\label{Q:decidability}
    Is there an algorithm which can compute (or even efficiently compute) the capacity?
\end{question}

Classically, the work of Shannon \cite{Sha48} shows that the capacity of a classical channel is given by optimising a certain entropic quantity — the mutual information, $I(X ; Y) \coloneqq H(X) + H(Y) - H(XY)$\footnote{Here, $H(X)$ is the Shannon entropy and $X, Y$ are classical random variables.}, over a single use of the channel. It immediately follows from the fact that $I(X ; Y)$ is \textit{additive} that there is an efficiently computable single-letter formula for the capacity of a classical channel.

In contrast, Question \ref{Q:decidability} is poorly understood in the quantum setting. While such formulae are known for some restricted families of quantum channels \cite{CRS08}, no such characterisation is known in general. The best known coding theorem \cite{Lloyd, Shor02, Devetak} expresses the quantum capacity in regularized form as
\[
Q(\Phi)=\lim_{n\to\infty} \frac{1}{n}Q^{(1)}(\Phi^{\otimes n}), 
\]
where $Q^{(1)}(\Phi)=\max_{\rho} I_c(\Phi,\rho)$ is the one-shot quantum capacity, defined via the coherent information \(I_c(\Phi,\rho)\) (see \cref{def: qcap}).

Superadditivity \cite{Has09} and superactivation \cite{SY08} already indicate that such a regularization is unavoidable. 
Stronger still, Cubitt et al.\ \cite{CEM+15} showed that for every $n\in\N$, there exists a quantum channel $\Phi$ such that $\frac{1}{n}Q^{(1)}(\Phi^{\otimes n}) < Q(\Phi)$. 
Taken together, these phenomena strongly suggest that the quantum capacity may be uncomputable \cite{P-ECG+25}. 
However, they do not directly address its computability. 
In particular, even a positive answer to the following question would not, on its own, resolve \cref{Q:decidability}.
\begin{question}\label{question2}
Does there exist a quantum channel $\Phi$ such that for every $n\in\N$,
\[
\frac{1}{n}Q^{(1)}(\Phi^{\otimes n}) < Q(\Phi)\,?
\]
\end{question}
While the existence of such an exotic channel would be remarkable, it falls short of resolving \cref{Q:decidability}. 
For example, it would not rule out the existence of a yet-undiscovered entropic quantity that provides a more direct method for computing $Q(\Phi)$.
The intuitive connection between \cref{Q:decidability} and \cref{question2} is nonetheless well founded, but the implication may actually run in the opposite direction. As we discuss in \Cref{subsec:parallel}, if the problem of determining whether $Q(\Phi)\geq \tfrac{1}{2}$ or $Q(\Phi)< \tfrac{1}{2}$ is suitably hard, in a computational sense, then \cref{question2} must admit a positive answer.

The connection between \cref{Q:decidability} and \cref{question2} draws strong parallels with undecidability phenomena in the study of non-local games \cite{FKM+25}. 
In that setting, progress toward undecidability was preceded by a sequence of intermediate hardness results. 
While we do not resolve \cref{Q:decidability}, we establish the first such intermediate hardness results toward this goal. 

\subsection{Main results}

 In this work we initiate a formal investigation of the computational hardness and computability of quantum channel capacities and give two hardness results in the spirit of answering Question \ref{Q:decidability}.  

\subsubsection{QMA hardness}
Our first main result is that estimating quantum capacities of quantum channels is QMA-hard.

\begin{theorem}[Informal version of \cref{cor:QMA-hard}]\label{thm:QMA}
    Given a quantum channel $\Phi$, it is \textsf{QMA}-hard to decide if the quantum capacity, $Q(
    \Phi)$, is $\geq \frac{3}{4}$ or $\leq \frac{1}{4}$, promised one of them holds. 
\end{theorem}

In fact, we establish a reduction from any \textsf{QMA}-complete language to the problem of estimating the quantum capacity of \emph{degradable} quantum channels, which then implies that it is QMA hard to estimate the quantum capacity of an arbitrary quantum channel. Since for degradable channels, the quantum capacity equals the one-shot quantum capacity \cite{CRS08}, our reduction also yields \textsf{QMA}-hardness of the one-shot quantum capacity for degradable channels. This is somewhat surprising: degradable channels are among the best understood classes of quantum channels, and the one-shot quantum capacity already has a simple single-letter formula, yet estimating one-shot quantum capacity of degradable channels remains hard even for quantum computers. 

While \cref{thm:QMA} does not rule out the possibility of a single-letter formula for the quantum capacity in full generality, it does show that such a formula, even when available, need not be efficiently computable (unless \textsf{P}$=$\textsf{QMA}).  

\subsubsection{Undecidability}
In the spirit of Question \ref{Q:decidability}, our second main contribution establishes an undecidability result for a quantum channel capacity. Specifically, we consider a restricted variant of the entanglement-assisted one-shot zero-error classical capacity of a classical-to-quantum (abbrev. c-q) channel, denoted $C^{(1)}_{0,\PME}(\Phi)$. We show that the threshold problem for this capacity is undecidable.

\begin{theorem}[Informal version of \Cref{theorem:RE}]\label{thm:RE-inf}
    Given a c-q channel $\Phi$ and an integer $t$, it is \textsf{RE}-hard to determine whether $C^{(1)}_{0,\PME}(\Phi)\geq \log(t)$. 
\end{theorem}

Operationally, the entanglement-assisted one-shot zero-error classical capacity $C^{(1)}_{0,E}(\Phi)$ of a quantum channel $\Phi$ is the maximum number of classical messages that can be transmitted perfectly through a single use of $\Phi$, assuming the sender and receiver share entanglement in advance. The restricted capacity, $C^{(1)}_{0,\PME}(\Phi)$, is defined similarly for c-q channels, except that the shared entangled state is required to be \emph{maximally entangled}, and decoding operations are projective measurement channels (see \Cref{def: max-EA}).

At a conceptual level, the restrictions appearing in the definition of  $C^{(1)}_{0,\PME}(\Phi)$ reflect the structure of perfect strategies of nonlocal games used in our reduction, which we explain in the technical overview below. Thus, \cref{thm:RE-inf} shows that undecidability already arises for a restricted notion of zero-error capacity for quantum channels.

\subsection{Technical contributions and proof ideas}

Next, we present an overview of the techniques used and highlight the intuitive idea behind our proofs. 

\subsubsection{Reduction from \textsf{QMA} problems to quantum capacity of channels}

\paragraph{Encoding channels as circuits.} 

To define our promise problem related to \cref{thm:QMA} we follow the approach of \cite{CM23} and consider a quantum channel to be encoded as the description of a quantum circuit $C$ which implements the channel $\Phi_C$. These circuits are considered to be comprised of unitary gates from some standard universal gate set, as well as non-unitary operations of partial trace and state preparation. Our problem is then as follows: 
\begin{quote}
    Given the description of a circuit $C$, decide if the corresponding channel $\Phi_C$ satisfies $Q(\Phi_C) \geq \frac{3}{4}$ or $Q(\Phi_C) \leq \frac{1}{4}$, given one of these two cases holds.
\end{quote}
To prove the \textsf{QMA}-hardness for the above decision problem, we appeal directly to the definition of the complexity class \textsf{QMA}. The starting point is to consider an equivalent characterisation in which a verifier, given an input $x$ from a \textsf{QMA} language, produces a classical description of the quantum circuit $V(x)$ (see \cref{fig:QMA}). This quantum circuit implements a quantum channel $\Phi_{V(x)}$ from a proof register to a single-qubit decision register $R$. The probability, $\text{Pr}(\text{verifier accepts }x \text{ with }\rho)$, that the verifier accepts $x$ with a proof state $\rho$ is then given by
\begin{equation}
    \langle 1 \vert \Phi_{V(x)}(\rho) \vert 1 \rangle = \text{Tr}\big(U_x^{\dagger}(\vert 1\rangle\langle 1 \vert\otimes I)U_x\rho    \big), \label{eq1}
\end{equation}
where $U_x$ is the Stinespring isometry for $\Phi_{V(x)}$. Suppose this QMA language has completeness $c$ and soundness $s$. If $x\in Y$, then there exists a proof state $\rho$ such that $\text{Pr}(\text{verifier accepts }x \text{ with }\rho)\geq c$; if $x\in N$, then $\text{Pr}(\text{verifier accepts }x \text{ with }\rho)\leq  s$ for all proof states $\rho$.

\paragraph{Intuitive approach of \textsf{QMA}-hardness.}
An intuitive approach for the reduction would be to choose two channels whose quantum capacities are well understood and exemplify the extremes: a channel $\Phi_{\mathcal{Y}}$ with full capacity, and a channel $\Phi_{\mathcal{N}}$ with zero capacity. We then consider the following construction. Fix a \textsf{QMA} language $(Y,N)$, and map each instance $x$ to a channel $\Phi_x$ defined as follows. $\Phi_x$ first applies the Stinespring dilation $U_x$ of the verification channel $\Phi_{V(x)}$ to the first register $R$, say of the input state $\rho$, and then measures $R$ in the computational basis. Conditioned on the first bit of decision register $R$, it applies $\Phi_{\mathcal{Y}}$ to $R'$, the second register of $\rho$ if $R$ is in the state $\vert 1 \rangle \langle 1 \vert$, and applies $\Phi_{\mathcal{N}}$ to $R'$ if $R$ is in the state $\vert 0 \rangle \langle 0 \vert$. At a heuristic level, this construction suggests the desired reduction: if $x \in Y$, then $\Phi_x$ behaves similarly to $\Phi_{\mathcal{Y}}$ and has large quantum capacity, while if $x \in N$, then $\Phi_x$ behaves similarly to $\Phi_{\mathcal{N}}$ and has small quantum capacity.

While this approach seems promising, the quantum capacity $Q(\Phi)$ is notoriously difficult to analyse, even for channels constructed from seemingly simple components. To address this, we introduce a variant of the well-known direct sum of quantum channels, which we call the \emph{symmetric direct sum}, and show that, together with an appropriate choice of $\Phi_{\mc{Y}}$ and $\Phi_{\mc{N}}$, the spirit of this approach can be made rigorous.

\paragraph{Symmetric direct sums of channels} The standard direct sum of channels $\Phi_0: B(\mathcal{H}_0) \rightarrow B(\mathcal{K}_0)$ and $\Phi_1: B(\mathcal{H}_1) \rightarrow B(\mathcal{K}_1)$ is given by \begin{equation*}\label{eq:DirectSum}
    \Phi_0 \oplus \Phi_1  \begin{pmatrix}
    X_0 & 0 \\ 0 & X_1
\end{pmatrix}  = \begin{pmatrix}
    \Phi_0(X_0) & 0 \\ 0 & \Phi_1(X_1)
\end{pmatrix}. \end{equation*} It is known that this defines a valid quantum channel $\Phi_0 \oplus \Phi_1: B(\mathcal{H}_0 \oplus \mathcal{H}_1) \rightarrow B(\mathcal{K}_0 \oplus \mathcal{K}_1)$ \cite{watrous2018theory}. More general direct sum constructions have been studied in the context of channel capacities, including partially coherent direct sums \cite{CG21} and generalised direct sums \cite{WZ25}.

Such constructions are not suitable for our purposes. First, the inputs to our promise problems are encoded as classical descriptions of quantum circuits. Although the direct sums above are natural at the level of linear-algebraic expressions, they do not interact cleanly with the circuit model. In particular, given circuits $C_0$ and $C_1$ implementing channels $\Phi_0$ and $\Phi_1$, it is not straightforward to construct a circuit implementing the direct sum $\Phi_0 \oplus \Phi_1$. One concrete obstruction is that the Hilbert spaces $\mathcal{H}_0 \oplus \mathcal{H}_1$ and $\mathcal{K}_0 \oplus \mathcal{K}_1$ need not have dimension a power of two, making their direct realisation within the circuit formalism unwieldy. The second, and more substantial, issue with the constructions above is that for our purposes we must combine two channels $\Phi_0$ and $\Phi_1$ with a common input space $R$ in such a way that the resulting channel still acts on $R$, rather than on the direct sum $R \oplus R$. To address this, in \cref{def:symm-dir-sum}, we introduce a new construction that generalises the direct sum of channels as well as the probabilistic direct sums of \cite{CM23}. 

Concisely, given channels $\Phi_0, \Phi_1: B(\mathcal{H}) \rightarrow B(\mathcal{K})$ and an isometry $U:\mathcal{H}'\rightarrow (\mathbb{C}^2)^{\otimes n}$ we define the symmetric direct sum as 
\begin{align*}
&\Phi_0 \oplus_{U} \Phi_1 (\sigma\otimes\rho) =
    \begin{pmatrix}
        \mathcal{F}_0(\sigma) \otimes \Phi_0 (\rho) & 0 \\
        0 & \mathcal{F}_1(\sigma) \otimes \Phi_1 (\rho)
    \end{pmatrix},
\end{align*}
where $\mathcal{F}_i(\sigma) = (\langle i \vert \otimes I)\mathcal{M}(U\sigma U^\dag)(\vert i \rangle \otimes I),i=0,1$, and $\mathcal{M}(\tau)=\sum_{x\in\{0,1\}^n}\vert x \rangle \langle x \vert \tau \vert x \rangle \langle x \vert$ is the channel corresponding to measurement in the computational basis. Extending linearly from elementary tensors to arbitrary states, it can be confirmed that $\Phi_0 \oplus_{U} \Phi_1$ is a well-defined quantum channel from $ B(\mathcal{H}'\otimes\mathcal{H}) \rightarrow B((\mathbb{C}^2)^{\otimes n}\otimes\mathcal{K})$.

As outlined in \cref{rem:cir-dir-sum}, given circuits implementing the channels \(\Phi_0\) and \(\Phi_1\), together with a circuit implementing the unitary $U$, one can construct, in polynomial time, a circuit for the channel $\Phi_0 \oplus_U \Phi_1$ as illustrated in \cref{fig:circuit}. The strength of the symmetric direct sum comes from two key properties that allow us to analyse $Q(\Phi_0 \oplus_U \Phi_1)$ through the one-shot quantum capacity for suitable choices of $\Phi_0$ and $\Phi_1$. First, the complement of the symmetric direct sum is equal to the symmetric direct sum of the individual complementary channels — a property that also holds for standard direct sum constructions. Second, due to the symmetry between the output and environment, it preserves degradability. The latter notably is a novel property of the symmetric direct sum not observed in previous constructions.

\begin{enumerate}[]
    \item \textit{Property 1.} $(\Phi_0 \oplus_U \Phi_1)^c = \Phi^c_0 \oplus_U \Phi^c_1$.
    \item \textit{Property 2.}  If $\Phi_0, \Phi_1$ are degradable channels, then $\Phi_0 \oplus_U \Phi_1$ is degradable, and consequently $Q(\Phi_0 \oplus_U \Phi_1) = Q^{(1)}(\Phi_0 \oplus_U \Phi_1)$.
\end{enumerate}

See \cref{lem:comp-direct-sum,lem:degr-dir-sum} for the precise statements and proofs of these properties. These observations are crucial for our approach in determining $Q(\Phi_0 \oplus_U \Phi_1)$ when $\Phi_0$ and $\Phi_1$ are chosen from the correct family of channels.

\paragraph{Role of degradable and anti-degradable channels}

In our proof of \cref{thm:QMA} via the symmetric direct sum construction, we select channels with a certain property such that the quantum capacity of the symmetric direct sum, $Q(\Phi_0 \oplus_U \Phi_1)$ can be expressed in terms of $Q(\Phi_0)$ and $Q(\Phi_1)$. Namely, we choose a channel that, depending on certain channel parameters, is degradable as well as anti-degradable (see \cref{sec:channels} for a definition of (anti-)degradability and properties of such channels). We then take the two extreme cases which lead to a channel with maximum capacity like the identity channel, and zero capacity like an anti-degradable channel. 

Intuitively, the property of degradability gives us that the coherent information, an entropic quantity that characterises the one-shot quantum capacity, is additive. Hence, the quantum capacity becomes exactly equal to the coherent information, and no longer requires a regularisation. The property of anti-degradability gives us that the quantum capacity is zero, as it allows the environment of the channel to simulate the state of the receiver (see \cref{sec:channels}). Specifically for our proof, we choose, $\mc{A}_{\eta}$, the qubit amplitude damping channel, and take $\mc{A}_0 = \text{id}$ the identity channel, and $\mc{A}_{1/2}$ an anti-degradable channel with no capacity as described in \cref{rem:amp-damp}.

Properties 1 and 2 allow us to compute the quantum capacity of any symmetric direct sum of  $\mathcal{A}_{1/2}$ and $\mathcal{A}_{0}$. For any isometry $U$, we show that
\begin{equation}
Q(\mathcal{A}_{1/2}\oplus_U \mathcal{A}_{0})=\max_{\rho} \text{Tr}\big(U^{\dagger}(\vert 1\rangle\langle 1 \vert\otimes I)U\rho    \big).\label{eq2}
\end{equation}
Given any QMA language $(Y,N)$, we map each instance $x\mapsto \Phi_x:= \mathcal{A}_{1/2}\oplus_{U_x} \mathcal{A}_{0}$, where $U_x$ is the Stinespring isometry of the verification channel $\Phi_{V(x)}$. Comparing \Cref{eq1,eq2}, we see that
\begin{equation}
    Q(\Phi_x)=\max_{\rho} \text{Pr}(\text{verifier accepts }x \text{ with }\rho). 
\end{equation}
 The mapping from $x$ to the description of the circuit that implements $\Phi_x$ is efficient, as illustrated in \Cref{fig:circuit}. This establishes a Karp reduction from any QMA-complete problem to the problem of estimating the quantum capacity of quantum channels.

\subsubsection{Undecidability via graph parameters}

\paragraph{Classical capacity and independence number} Let $G=(V,E)$ be an undirected graph. The \emph{independence number} $\alpha(G)$ is the maximum size of a set of pairwise non-adjacent vertices. Given a classical channel $\Phi$ with confusability graph $G$, it is well known that the one-shot zero-error capacity satisfies 
\begin{equation}
    C^{(1)}_0(\Phi)=\log(\alpha(G)).\label{eq3}
\end{equation}
Moreover, for any graph $G$, one can efficiently construct a classical channel $\Phi_G$ such that $G$ is the confusability graph of $\Phi_G$. Thus, computing the one-shot zero-error capacity of a classical channel is at least as hard as computing the independence number of a graph. Since computing $\alpha(G)$ is \textsf{NP}-complete, it follows that computing $C^{(1)}_0(\Phi)$ is \textsf{NP}-hard. 

Our undecidability result, \Cref{thm:RE-inf}, follows from a similar complexity-theoretic paradigm. Namely, we establish a reduction from \emph{quantum independence number}, which is known to be uncomputable, to the restrict variant of the entanglement-assisted one-shot zero-error capacity of a quantum channel.

\paragraph{Quantum independence number}  An operational characterisation of the independence number $\alpha(G)$ is given by the \emph{independent set game} $\mathrm{IS}(G,t)$, introduced in \cite{MRV15}. In this game, two non-communicating players aim to convince a verifier that $G$ has an independent set of size at least $t$. The game $\mathrm{IS}(G,t)$ has a perfect classical strategy if and only if $\alpha(G)\geq t$. Allowing the players to share (finite-dimensional) entanglement leads to the \emph{quantum independence number} $\alpha_q(G)$, defined as the largest $t$ for which $\mathrm{IS}(G,t)$ admits a perfect \emph{quantum strategy}. The recent result $\textsf{MIP}^{\ast}=\textsf{RE}$~\cite{JNV+21} implies that it is \textsf{RE}-hard to determine if a given independent set game has a perfect quantum strategy. As a result, given any graph $G$ and a positive integer $t$, it is \textsf{RE}-hard to determine whether $\alpha_q(G)\geq t$. 

In \Cref{sec: RE-hard}, we establish the following quantisation of \Cref{eq3}.
\begin{theorem}[Resated from \Cref{theorem:main=}]\label{theorem:main=inf}
    There is a computable mapping from graphs $G$ to c-q channels $\Phi_G$ such that
    \begin{equation}
    C^{(1)}_{0,\PME}(\Phi_G)=\log (\alpha_q(G)).\label{eq4}
\end{equation}
\end{theorem}
\Cref{theorem:main=inf} yields a reduction from the problem of computing the quantum independence number of graphs to the problem of computing $C^{(1)}_{0,\PME}$ of c-q channels. As a consequence, given any c-q channel $\Phi$ and a positive integer $t$, determine whether $C^{(1)}_{0,\PME}(\Phi)\geq \log(t)$ is \textsf{RE}-hard.

\paragraph{Classical-to-quantum channels} A c-q channel $\Phi$ is specified by a finite alphabet $V$ and a family of output densities $\{\rho_v\}_{v\in V}$. Its confusability graph\footnote{In fact, this ``classical" graph also matches the noncommutative graph associated with the quantum channel $\Phi$.} $G$ is the graph with vertices $V$, in which distinct vertives $v,w\in V$ are adjacent whenever $\Tr(\rho_v\rho_w)\neq 0$. Conversely, given any graph $G$, we can construct a c-q channel $\Phi_G$ whose confusability graph is exactly $G$. This is precisely the mapping from graphs to channels used in \cref{theorem:main=inf}. 

The key point is that, for the channels $\Phi_G$, transmitting $t$ bits in our restricted entanglement-assisted setting is equivalent to winning the independent set game $\mathrm{IS}(G,t)$ perfectly. More precisely, we establish a one-to-one correspondence between $(\ket{\psi},\mc{E}_i,1\leq i\leq t)$, the state and encoding channels that achieve $C^{(1)}_{0,\PME}(\Phi_G)$, and $(\ket{\psi},\{P^i_v\}_{v\in V},1\leq i\leq t)$, the state and measurements that win with certainty for $\mathrm{IS}(G,t)$. This correspondence yields \cref{theorem:main=inf}. 

The independent set game $\mathrm{IS}(G,t)$ belongs to a broader class of \emph{synchronous games}, whose perfect strategies are well-understood and admit a particular form: the shared state is maximally entangled, and the measurements are projection-valued. These features are precisely what restrict our results to only pertain to the capacity $C^{(1)}_{0,\PME}$ as opposed to $C^{(1)}_{0,E}$.

\subsection{Parallel developments and conjectures in undecidability }\label{subsec:parallel}
A recurring theme in recent quantum information theory is that quantities with clear operational meaning can nonetheless exhibit extreme computational hardness. The most developed examples are values of nonlocal games: their complexity theory has progressed from initial hardness results~\cite{IV12,NW19} to full undecidability~\cite{Slo19,Slo20,JNV+21,MNY22,NZ23,noise25,CE25,TV25,FKM+25,Lin25,MSSV25,CMPS25}. Quantum channel capacities are another natural family of operational quantities, but their complexity landscape remains sparsely understood. In this light, our results may therefore be seen as early steps toward a complexity-theoretic understanding of quantum channel capacities paralleling the development that took place for nonlocal games.

In the study of nonlocal games, the \emph{quantum value} $\omega_q(\mc{G})$ is defined as the supremal winning probability over all finite-dimensional quantum strategies, while the \emph{commuting operator value} $\omega_{co}(\mc{G})$ is the supremal winning probability over all strategies in the commuting operator model. Both values admit convergent hierarchies with computable values at each level. For $\omega_q(\mc{G})$, one may restrict to strategies of local dimension $n$, obtaining optimal values $\omega_{\FD}^{(n)}(\mc{G})$ satisfying
\begin{equation*}
   \omega_{\FD}^{(1)}(\mc{G})\leq \omega_{\FD}^{(2)}(\mc{G})\leq \cdots,\quad  \lim_{n\rightarrow \infty}\omega_{\FD}^{(n)}(\mc{G})=\omega_q(\mc{G}).
\end{equation*}
For $\omega_{co}(\mc{G})$, the NPA hierarchy \cite{NPA08} provides a sequence of semidefinite relaxations with optimal values $\omega_{\NPA}^{(n)}(\mc{G})$ satisfying
\begin{equation*}
    \omega_{\NPA}^{(1)}(\mc{G}) \geq \omega_{\NPA}^{(2)}(\mc{G})\geq \cdots,\quad \lim_{n\rightarrow\infty}\omega_{\NPA}^{(n)}(\mc{G})=\omega_{co}(\mc{G}).
\end{equation*}
Since both hierarchies converge (from below to $\omega_q$, and from above to $\omega_{co}$), a natural question is whether the game values are always attained at some finite level.
\begin{question}\label{question3}
    Does there exist a nonlocal game $\mc{G}$ such that for every $n\in\N$, $\omega_{\FD}^{(n)}(\mc{G})< \omega_q(\mc{G})$ (resp. $\omega_{\NPA}^{(n)}(\mc{G})> \omega_{co}(\mc{G})$)?
\end{question}
Earlier results gave strong evidence in this direction. Slofstra~\cite{Slo11}, and Harrow, Natarajan, and Wu~\cite{HNW17} proved that for every level $n$ there exists a nonlocal game $\mc{G}$ such that $\omega_{\FD}^{(n)}(\mc{G})< \omega_q(\mc{G})$, and (respectively) a game $\mc{G}$ such that $\omega_{\NPA}^{(n)}(\mc{G})> \omega_{co}(\mc{G})$. Although these results do not answer \Cref{question3}, they suggest that attaining these values may require unbounded levels.

The eventual resolution of \Cref{question3} is indeed closely related to undecidability. Slofstra~\cite{Slo19} answered \Cref{question3} affirmatively for $\omega_q$ by constructing a linear system game $\mc{G}$ such that $\omega_q(\mc{G})=1$ but $\omega_{\FD}^{(n)}(\mc{G})<1$ for all $n\in\N$. Using similar techniques, he also showed that, given a nonlocal game $\mc{G}$, it is undecidable to determine if there exists $n$ such that $\omega_{\FD}^{(n)}(\mc{G})=1$.  More recently, \cite{FKM+25} answered \cref{question3} affirmatively for $\omega_{co}$. They prove that the decision problem
\begin{quote}
    Given a nonlocal game $\mc{G}$, is $\omega_{co}(\mc{G})>1/2$ or $\leq 1/2$.
\end{quote}
is \textsf{RE}-hard. If the NPA hierarchy always attains the commuting operator value at some finite level, this problem would lie in \textsf{coRE}, contradicting \textsf{RE}-hardness.

By the superadditivity of the one-shot quantum capacity $Q^{(1)}$, the regularised formula for quantum capacity also gives rise to a natural convergent hierarchy. 
\begin{equation*}
    Q_{reg}^{(1)}(\Phi)\leq Q_{reg}^{(2)}(\Phi)\leq \cdots,\quad \lim_{n\rightarrow \infty} Q_{reg}^{(n)}(\Phi) = Q(\Phi),
\end{equation*}
where $Q_{reg}^{(n)}(\Phi):=\frac{Q^{(1)}(\Phi^{\otimes 2^n})}{2^n}$. Our \cref{question2}, which asks if there exists a quantum channel $\Phi$ such that for every $n\in\N$, $Q_{reg}^{(n)}(\Phi) < Q(\Phi)$, is in direct analogy with \Cref{question3}.

As in the nonlocal-game setting, there are already partial results pointing in answering \Cref{question2}.
Cubitt, Elkouss, Matthews, Ozols, Pérez-García, and Strelchuk in \cite{CEM+15} showed that for every $n\in\N$, there exists a quantum channel $\Phi$ such that $Q_{reg}^{(n)}(\Phi) < Q(\Phi)$. For a full resolution of \cref{question2}, one needs to alternate the quantifiers showing there exists a channel $\Phi$ such that for all $n \in \N$, $Q_{reg}^{(n)}(\Phi) < Q(\Phi)$.

Unlike in the setting of non-local games however, at each level the function $Q_{reg}^{(n)}(\Phi)$ may not be computable: computing $Q_{reg}^{(n)}(\Phi)$ involves maximizing a continuous function on a compact set, which in general can be accomplished in the second level of the arithmetical hierarchy $\Pi_2^0$. Nonetheless, determining the computational hardness of the decision problem
\begin{equation*}\tag{Qcap}
    \text{Given a quantum channel $\Phi$, determine if $Q(\Phi)\geq \frac{1}{2}$ or $Q(\Phi)< \frac{1}{2}$.}
\end{equation*}
 can still potentially resolve \cref{question2}. In fact, assuming a negative answer to \cref{question2}, then (QCap) would be in $\Sigma_3^0$. A positive answer to \cref{question2} would then follow from the following conjecture:
\begin{conjecture}\label{conjecture1}
   (QCap) is $\Pi_3^0$-hard.
\end{conjecture}
To the best of our knowledge, the smallest upper bound on (QCap) is the 4th level of the arithmetical hierarchy $\Pi_4^0$.

We do note that $Q_{reg}^{(n)}(\Phi)$ can be approximated to within arbitrary precision, so the gapped version
\begin{equation*}\tag{GapQcap}
    \text{Given a quantum channel $\Phi$, is $Q(\Phi)\geq 3/4$ or $\leq 1/4$.}
\end{equation*}
is in \textsf{RE}.
 We conjecture that
\begin{conjecture}  
 (GapQcap) is \textsf{RE}-complete.
\end{conjecture}

\subsection{Discussion and outlook}\label{subsec:outlook}

In this work, we initiate a formal study into the computational hardness of channel capacities for quantum channels. A natural long-term objective here is to understand whether the quantum capacity is uncomputable for general channels. This objective is anticipated to be extremely challenging, and in the discussion below we highlight both the insights gained toward this goal and the respects in which the present results fall short.

\paragraph{Encoding of the input channels} First, our result in \cref{thm:QMA} considers a formulation in which the input is a quantum circuit $C$ implementing a channel $\Phi_C$. In this setting, we rule out the existence of an efficient algorithm that, given $C$, approximates the quantum capacity $Q(\Phi_C)$ even within a constant error. On the other hand, in several areas of quantum information theory, quantum channels are often studied via explicit mathematical descriptions or formulas, rather than through circuit representations. Furthermore, succinct input representations of the kind used here are well known to increase the apparent computational complexity of a problem. However, if the quantum capacity is indeed uncomputable, then such differences in how channels are specified become largely immaterial. In particular,  we show that the capacity $C^{(1)}_{0,\PME}$ of a c-q channel is uncomputable. So in \cref{thm:RE-inf} we need not specify if a channel is encoded as a quantum circuit, or via some other computable encoding using mathematical formulas, such as descriptions of Kraus operators. 

\paragraph{Entanglement-assisted zero-error capacities and graph parameters} Since Shannon's introduction of zero-error communication through confusability graphs, graph parameters have played a central role in the study of one-shot zero-error capacities. In the entanglement-assisted setting, this connection was developed further in the work of Cubitt, Leung, Matthews, and Winter~\cite{CLMW10}, who showed that shared entanglement can increase the one-shot zero-error capacity of a classical channel, and in the work of Beigi~\cite{Bei10}, who related the entanglement-assisted zero-error capacity to the Lov\'asz $\vartheta$-function.

Particularly relevant to us is the latter work of Man\v{c}inska and Roberson on quantum graph homomorphisms~\cite{MR16}. They introduced the quantum independence number, and showed that it is closely related to the entanglement-assisted one-shot zero-error capacity $C^{(1)}_{0,E}$ for classical channels. More precisely, if a classical channel $\mc{N}$ has confusability graph $G$, then their Theorem 5.1 shows that $C^{(1)}_{0,E}(\mc{N})$ is at least $\log(\alpha_q(G))$, with equality whenever the capacity can be achieved using projective measurements on a maximally entangled state. Our \cref{theorem:main=inf} can be viewed as taking this idea one step further by considering quantum channels. In both \cite{MR16} and our setting, the relevant protocols come from perfect quantum strategies for the independent set game, and therefore naturally involve maximally entangled states and projective measurements. \cite{MR16} already asked whether (for classical channels) such restrictions can always be removed. For general quantum channels, however, later work of Stahlke~\cite{Sta16} shows that one can not expect this in full generality: there exists quantum channel $\Phi$ such that only using a non-maximally entangled state can achieve $C^{(1)}_{0,E}(\Phi)$. While the undecidability of the entanglement-assisted capacity remains open, the discussion here naturally motivates considering maximal entanglement and projective measurements leading to the special case, $C^{(1)}_{0,\PME}$.

\paragraph{Insights to Shannon capacity} The Shannon capacity \cite{Sha56} as we know is a central concept bridging zero-error information theory and graph theory, the computability of which has remained notoriously difficult (see \cite{LS25} for a review).
The beautiful work of Lovász shows that the Lovász $\vartheta$-function can be computed in polynomial time in $n$ (size of the adjacency matrix), and is an upper bound on the Shannon capacity, whose computation requires the computation of an infinite series of the independence number, which is a known \textsf{NP}-hard problem \cite{Lov79}. Moreover, even fundamental questions surrounding the Shannon capacity remain answered, the best-known being the determination of the Shannon capacity of odd cycle
graphs of length greater than 5 \cite{Lov79,dBBZ24,PS19,Boh04}. Such evidence for extreme uncomputability of even the Shannon capacity unravels the elusive nature of the entanglement-assisted case in full generality.

\paragraph{Comparison with prior works} Our \textsf{QMA}-hardness result of the quantum capacity begs a comparison with one of the results of Beigi and Shor ~\cite{BS07}, which shows that computing the classical (Holevo) capacity of a quantum channel is \textsf{NP}-complete. Notably, to the best of out knowledge, it is the only other significant complexity-theoretic lower bound known for quantum channel capacities. This general lack of formal hardness results highlights how fundamental questions about complexity of quantum channel capacities, including \Cref{Q:decidability}, have so far remained largely unexplored.

\subsection{Outline} \cref{sec: intro} is the introduction. We present the main results \cref{thm:main-QMA-hardness} and \cref{cor:QMA-hard} on the \textsf{QMA}-hardness of computing $Q(\Phi)$ in \cref{sec: QMA-hard} for a quantum channel $\Phi$, and \cref{theorem:main=} and \cref{theorem:RE} on the \textsf{RE}-hardness of computing $C^{(1)}_{0, \PME}(\Phi)$ in \cref{sec: RE-hard}. We review relevant background to follow the proofs of these results in \cref{sec: prelims}, namely quantum channel capacities and circuits, aspects of computational complexity, and the quantum independence number of graphs. 

\paragraph{Acknowledgements} We thank Eric Culf, Omar Fawzi, Laura Man{\v{c}}inska, and William Slofstra for helpful discussions on the complexity of quantum channel capacities, circuit and direct sum constructions of quantum channels, and undecidability of nonlocal games admitting a perfect quantum strategy. AB thanks Anne Broadbent for invaluable support. We acknowledge the support of the Natural Sciences and Engineering Research Council of Canada (NSERC)(ALLRP-578455-2022), the Air Force Office of Scientific Research under award number FA9550-20-1-0375 and of the Canada Research Chairs Program (CRC-2023-00173). AM acknowledges the support of NSERC DG 2024-06049. YZ~is supported by VILLUM FONDEN via QMATH Centre of Excellence grant number 10059 and Villum Young Investigator grant number 37532.

\section{Preliminaries} \label{sec: prelims}

\paragraph{Notation} For $n\in\N$, write $[n]=\{1,2,\ldots,n\}$. For a function $f:\N\rightarrow[0,1]$, write $f:\N\rightarrow(0,1)_{\exp}$ to mean that there exist $N,k>0$ such that $2^{-n^k}<f(n)<1-2^{-n^k}$ for all $n\geq N$. We write $\log$ for the base-$2$ logarithm. We denote registers by uppercase Latin letters $A,B,C,\ldots$; and we denote Hilbert spaces by uppercase script letters $\mc{H},\mc{K},\mc{L},\ldots$. We always assume registers are finite sets and Hilbert spaces are finite-dimensional. We denote an independent copy of a register $A$ by $A'$. Given a register $A$, the Hilbert space spanned by $A$ is $\mc{H}_A=\spn\!\!\set*{\ket{a}}{a\in A}\cong\C^{|A|}$. We indicate that an operator or vector is on register $A$ with a subscript $A$, omitting when clear from context. Given two registers $A$ and $B$, we write $AB$ for their cartesian product, and treat the isomorphism $\mc{H}_{AB}\cong\mc{H}_A\otimes\mc{H}_B$ implicitly. Given finite-dimensional Hilbert spaces $\mc{H}$ and $\mc{K}$, we write $B(\mc{H},\mc{K})$ for the set of all linear operators $\mc{H}\rightarrow\mc{K}$, $B(\mc{H})=B(\mc{H},\mc{H})$, $\mc{U}(\mc{H})\subseteq B(\mc{H})$ for the subset of unitary operators, and $D(\mc{H})\subseteq B(\mc{H})$ for the subset of density operators where $D(\mc{H}) \coloneqq \{\rho \in B(\mc{H}), \rho \geq 0, \Tr{\rho} = 1\}$. 
Write $\mc{U}(d)=\mc{U}(\C^{d})$.
An operator $A \in B(\mc{H})$ is positive, denoted by $A \geq 0$, if $A = \sqrt{A^{\dagger}A}$, where $(\cdot)^{\dagger}$ represents the Hermitian conjugate. We denote by $\text{id}_R$ the identity map on $B(\mc{H}_R)$. We write $\Tr$ for the trace on $B(\mc{H})$. On $B(\mc{H}_{AB})$, we write the partial trace $\Tr_{B}=\id\otimes\Tr$. For $\rho_{AB}\in B(\mc{H}_{AB})$, write $\rho_A=\Tr_B(\rho_{AB})$. We denote the $L_1$-norm by $\norm{\cdot}_1$ and the trace norm by $\frac{1}{2}\norm{\cdot}_1$. We denote the operator norm by $\norm{\cdot}$. The identity operator is denoted by $I$. We denote the canonical maximally-entangled state $\ket{\phi^+}_{AA'}=\frac{1}{\sqrt{|A|}}\sum_{a\in A}\ket{a}\otimes\ket{a}\in\mc{H}_{AA'}$. We write the maximally-mixed state on a register $A$ as $\omega_A=\frac{1}{|A|}\sum_{a\in A}\ketbra{a}\in D(\mc{H}_A)$. 
A positive-operator-valued measurement (POVM) is a finite set of positive operators $\{P_i\}_{i\in I}$ such that $\sum_iP_i=\mds{1}$, and a projection-valued measurement (PVM) is a POVM where all the elements are projectors. A quantum channel is a completely positive trace-preserving (CPTP) map $\Phi:B(\mc{H})\rightarrow B(\mc{K})$. We denote the Choi-Jamio\l{}kowski isomorphism $J:B(B(\mc{H}_A),B(\mc{H}_B))\rightarrow B(\mc{H}_{AB})$, $J(\Phi)=(\id\otimes\Phi)(\ketbra{\phi^+}_{AA'})$. Note that if is $\Phi$ is a quantum channel, $J(\Phi)\in D(\mc{H}_{AB})$, called the Choi state.
The Kraus representation of a quantum channel $\Phi: B(\mc{H}) \to B(\mc{H})$ is $\sum \limits_{i = 0}^{d - 1} \Phi(\rho) = A_i \rho A_i^{\dagger}$ where $A_i \in B(\mc{H})$ are the Kraus operators such that $\sum \limits_{i = 0 }^{d - 1} A_i^{\dagger} A_i = I,$ and $d = \text{dim}(\mc{H})$. For $\Phi: B(\mc{H}_A) \to B(\mc{H}_B)$, by Stinespring's dilation theorem, there exists an isometry $V: A \to BE$ known as the Stinespring isometry such that $\Phi (\rho) = \Tr_E (V \rho V^{\dagger})$. The complementary channel $\Phi^c: B(\mc{H}_A) \to B(\mc{H}_E)$ is given by $\Phi^c (\rho) = \Tr_B (V \rho V^{\dagger}).$ 
Note that the Kraus representation is not unique. However, for a channel $\Phi$ we can always find its minimal Stinespring representation and corresponding minimal Kraus representation. Every Stinespring representation is equivalent to a minimal one via an isometry on the environment denoted by register $E$. The diamond norm of a linear map $\Phi: B(\mc{H}_A) \to B(\mc{H}_B)$ is defined as $\norm{\Phi}_{\diamond} \coloneqq \sup \limits_{\norm{M}_1 \leq 1} \norm{\Phi (M_{RA})}_1$, with supremum over all $M \in B(\mc{H}_{RA})$ with $d_R = d_A$ and $\norm{M}_1 \leq 1.$ 

\subsection{Quantum channel capacities and circuits} \label{sec:channels}

In this section, we introduce various capacities associated with quantum channels. In the process, we introduce the relevant entropic quantities and their properties, and the notion of circuit construction of quantum channels.

We focus on the quantum capacity and the entanglement-assisted zero-error capacity in this work. We recall that the von Neumann entropy of a quantum state $\rho$ is defined as $S(\rho) \coloneqq - \Tr(\rho \log \rho).$ The coherent information \cite{Bar98}, or equivalently, the one-shot quantum capacity is an entropic quantity denoted by $Q^{(1)}(\Phi)$. It characterises the amount of quantum information that can be transmitted over a quantum channel $\Phi$ (in a single use).

\begin{definition}[Quantum capacity] \label{def: qcap}
     The \textbf{quantum capacity} of a channel $\Phi$ is given by 
     \begin{equation}
     Q(\Phi) \coloneqq \lim \limits_{n \to \infty} \frac{Q^{(1)}\left(\Phi^{\otimes n}\right)}{n}
     \end{equation}
     where $$Q^{(1)} (\Phi) \coloneqq \max \limits_{\rho} I_c (\Phi, \rho),$$ and $I_c \left(\Phi, \rho\right) = S\left[(\Phi(\rho)\right] -~S\left[\Phi^c(\rho)\right]$ is the entropy exchanged between the channel $\Phi$ and its environment.
\end{definition}
A channel $\mc{N}$ is said to be \textbf{degradable} if there exists another channel $\mc{D}$, known as the \textit{degrading map}, such that $\mc{D} \circ \mc{N} = \mc{N}^c.$ Similarly, a channel $\mc{M}$ is said to be \textbf{anti-degradable} if there exists an \textit{anti-degrading map} $\mc{L}$ such that $\mc{L} \circ \mc{M}^c = \mc{M}.$ Intuitively, a degradable channel allows its receiver (Bob) to simulate the state of the environment (Eve) for any input supplied by the sender (Alice). Similarly, an anti-degradable channel allows the environment to simulate the receiver's state for any input. For degradable channels, the quantum capacity is additive and hence $Q(\mc{N}) = Q^{(1)}(\mc{N}).$ For anti-degradable channels, the quantum capacity $Q(\mc{M}) = 0.$ Let $\text{id}_2$ be the identity channel on $B(\C^2)$. Since this channel sends the entire input to the receiver, the environment does not receive any information. It is easy to see that the action of the complementary channel $\text{id}^c_2 (\rho) = \Tr(\rho)$, and hence $I_c(\text{id}_2, \rho) = S(\rho)$. This implies that the coherent information $Q^{(1)}(\text{id}_2) = \max \limits_{\rho} S(\rho) = \log 2 = 1.$

The notion of \textbf{entanglement assistance} means that prior to communication through a noisy channel, the sender and receiver are allowed to share an entangled state. Sharing of prior entanglement has been shown to increase the capacity in certain cases, as well as to have no effect in others, see \cite{watrous2018theory} for a summary. Indeed, determining the information transmission capacities of quantum channels is still a fundamentally unresolved question. The \textbf{entanglement-assisted one-shot zero-error classical capacity}, denoted by $C^{(1)}_{0,E}(\Phi)$, is the number of bits perfectly transmitted in a single use of the channel $\Phi$ in the presence of preshared entanglement (not necessarily maximal) between the sender and receiver.

\begin{definition}[Entanglement-assisted one-shot zero-error classical capacity] \label{def: EA-cap}
   The entanglement-assisted one-shot zero-error classical capacity of a CPTP map $\Phi: B(\mc{H}_A) \to B(\mc{H}_B)$ is defined as

\begin{equation}
    C^{(1)}_{0,E}(\Phi) \coloneqq \sup \limits_{\ket{\psi}, M} \log \lvert M \rvert, 
\end{equation}
where the supremum is over all pure bipartite (entangled) states $\ket{\psi} \in D(\mc{H}_{{A_0}} \otimes \mc{H}_{{B_0}})$ and collections $M$ of CPTP maps $\{\mc{E}_m : B(\mc{H}_{A_0}) \to B(\mc{H}_A)\}_{m = 1}^{|M|}$ such that $$\forall~ m \neq m' : (\Phi \circ \mc{E}_m \otimes \text{id}_{B_0}) (\ket{\psi}) \perp (\Phi \circ \mc{E}_{m'} \otimes \text{id}_{B_0}) (\ket{\psi}).$$ 
\end{definition}

If we restrict the shared state between the sender and receiver to only be a maximally entangled state $\phi^+_{AA'} = \frac{1}{\sqrt{|A|}}\sum_{a\in A}\ket{a}\otimes\ket{a}\in\mc{H}_{AA'}$, then we can define the corresponding \textbf{maximal-entanglement-assisted} one-shot zero-error classical capacity. Along with another condition for projective measurements, we obtain the following capacity notion relevant for our work.

\begin{definition}[$\PME$-assisted one-shot zero-error classical capacity] \label{def: max-EA}
   The $\PME$-assisted one-shot zero-error classical capacity of a classical-to-quantum channel $\Phi: B(\C^{\Sigma}) \to B(\mc{H}_B)$ with the input alphabet $\Sigma$ is defined as
\begin{equation}
    C^{(1)}_{0,\PME}(\Phi) \coloneqq \sup \limits_{\phi^+, M} \log \lvert M \rvert, 
\end{equation}
where the supremum is over all pairs $(\phi^+, M)$ satisfying
\begin{itemize}
    \item $\phi^+ \in D(\mc{H}_{A_0} \otimes \mc{H}_{A_0})$ is a maximally entangled state, and
    \item $M$ consists of  quantum-to-classical channels $\{\mc{E}_m : B(\mc{H}_{A_0}) \to B(\C^{\Sigma})\}_{m = 1}^{|M|}\}$ with the output alphabet $\Sigma$ such that $\mc{E}_m(\rho)=\sum_{a\in\Sigma}\Tr(P^m_a\rho)\ketbra{a}{a}$ for some PVMs $\{P^m_a:a\in\Sigma\},1\leq m\leq \abs{M}$,
\end{itemize}
and $$\forall~ m \neq m' : (\Phi \circ \mc{E}_m \otimes \text{id}_{A_0}) (\phi^+) \perp (\Phi \circ \mc{E}_{m'} \otimes \text{id}_{A_0}) (\phi^+).$$
\end{definition}

\begin{remark} \label{rem: gates}
    The action of a quantum channel may also be represented as a quantum circuit. The circuit model of quantum computation describes unitary quantum operations (channels) wherein a universal gate set $S$ such as the Clifford$+ T$ gate set $\{CNOT, H, T\}$ can approximate any unitary to arbitrary precision by composing finitely many tensors of the elementary gates. To completely describe any quantum channel, two non-unitary operations are also required: the partial trace $\Tr : B(\mbb{C}^2) \to \mbb{C}$ on one qubit, and single-qubit state preparation in the computational basis, given by the mapping $\C \to B(\mbb{C}^2), \lambda \mapsto \lambda \ketbra{0}.$ Due to the Stinespring dilation theorem, the composite gate set $S \cup \{\Tr, \ketbra{0}\}$ can approximate any quantum channel to arbitrary precision. The Stinespring dilation also enables constructing the circuit of a quantum channel as a sequence of steps: single-qubit preparation, followed by the action of unitary gates, and finally a partial trace. In this work, we will assume that our quantum circuits for channels are always in this canonical form. See \cite{CM23} for more on the circuit construction of channels.
\end{remark}

\subsection{Complexity theory} \label{sec: cmplx-th}

In this work, we study decision problems related to the computability of various channel capacities. Notably, we show that computing the quantum capacity of a quantum channel is \textsf{QMA}-hard. and computing the entanglement-assisted zero-error classical capacity of a classical-quantum channel restricted to maximal entanglement is $\textsf{RE}$-hard. We formally introduce the exact meaning of such a statement about hardness of computing in this section.

A \textbf{Turing machine} is a tuple $M = (k, \Sigma, S, \delta)$ where $k \in \mbb{Z}_{\geq 0}$ is called the number of work tapes, $\Sigma$ is a finite set called the set of symbols containing a blank symbol $\Box$, $S$ is a finite set called the set of states containing the initial state \texttt{start} and the halt state $\texttt{halt},$ and $\delta: S \times \Sigma^{k + 2} \to S \times \Sigma^{k + 2} \times \{-1, 0, 1\}^{k + 2}$ is a function called the transition function. A Turing machine is an abstraction of a computation. A function $f$ is \textbf{computable} if there is a Turing machine $M$ such that $f(x) = M(x)$ for all $x$. A \textbf{language} is a subset $L \subseteq \{0, 1\}^{\ast}$ where $\Sigma^{\ast}$ denotes the set of words in $\Sigma.$ A Turing machine \textbf{decides} a language $L$ if $M(x) = 1$ when $x \in L$ and $M(x) = 0$ when $x \notin L.$ A \textbf{promise problem}
$P$ is a pair of disjoint subsets $Y(P), N(P) \subseteq \{0, 1\}^{\ast}.$
The elements of $Y(P) \cup N(P)$ are called the \textbf{instances} of $P$, and the elements of $Y(P)$ are the \textbf{yes instances} and the elements of $N(P)$ are the \textbf{no instances}. We write $x \in P$ to mean $x \in Y(P) \cup N(P)$. Every language $L$ corresponds to the promise problem $(L, L^c)$, so we exclusively work with promise problems in the following. A Turing machine $M$ decides $P$ if $M(x) = 1$ for all $x \in Y(P)$ and $M(x) = 0$ for all $x \in N(P)$. A \textbf{complexity class} is a collection of promise problems. A problem $P$ is \textbf{hard} for a complexity class \textsf{C} or \textsf{C}-hard, if every problem in \textsf{C} reduces to $P.$ The problem $P$ is said to be $\textsf{C}$-\textbf{complete} if it is \textsf{C}-hard and $P \in \textsf{C}.$

The complexity class \textsf{QMA} is the quantum analogue of \textsf{NP} and, informally, consists of those promise problems that admit an efficiently verifiable quantum witness. A promise problem $(Y,N)$ is called \textsf{QMA}-hard if, for every promise problem $(Y',N')$ in \textsf{QMA}, there exists a polynomial-time computable function $f : \{0,1\}^* \to \{0,1\}^*$ such that $f(Y') \subseteq Y$ and $f(N') \subseteq N$. A promise problem is called \textsf{QMA}-complete if it is both \textsf{QMA}-hard and in \textsf{QMA}. The canonical example of a \textsf{QMA}-complete is the local Hamiltonian problem \cite{KSV02}.

We define the class \textsf{QMA} as follows.

\begin{definition}[QMA] \label{def: class-qma}
 Let $c, s : \mbb{N} \to (0, 1).$ A promise problem $Y, N \subseteq \{0, 1\}^{\ast}$ is in $\textsf{QMA}_{c, s}$ if there exist polynomials $p, r, l : \mbb{N} \to \mbb{N}$ and a Turing machine $V$ with one input tape and one output tape such that
\begin{itemize}
    \item For all $x \in \{0, 1\}^{\ast}$, $V$ halts on input $x$ in $r(|x|)$ steps and outputs the description of a quantum circuit $V(x)$ (see \cref{fig:QMA}) \begin{figure}[h]
    \centering
    \includegraphics[width=0.7\linewidth]{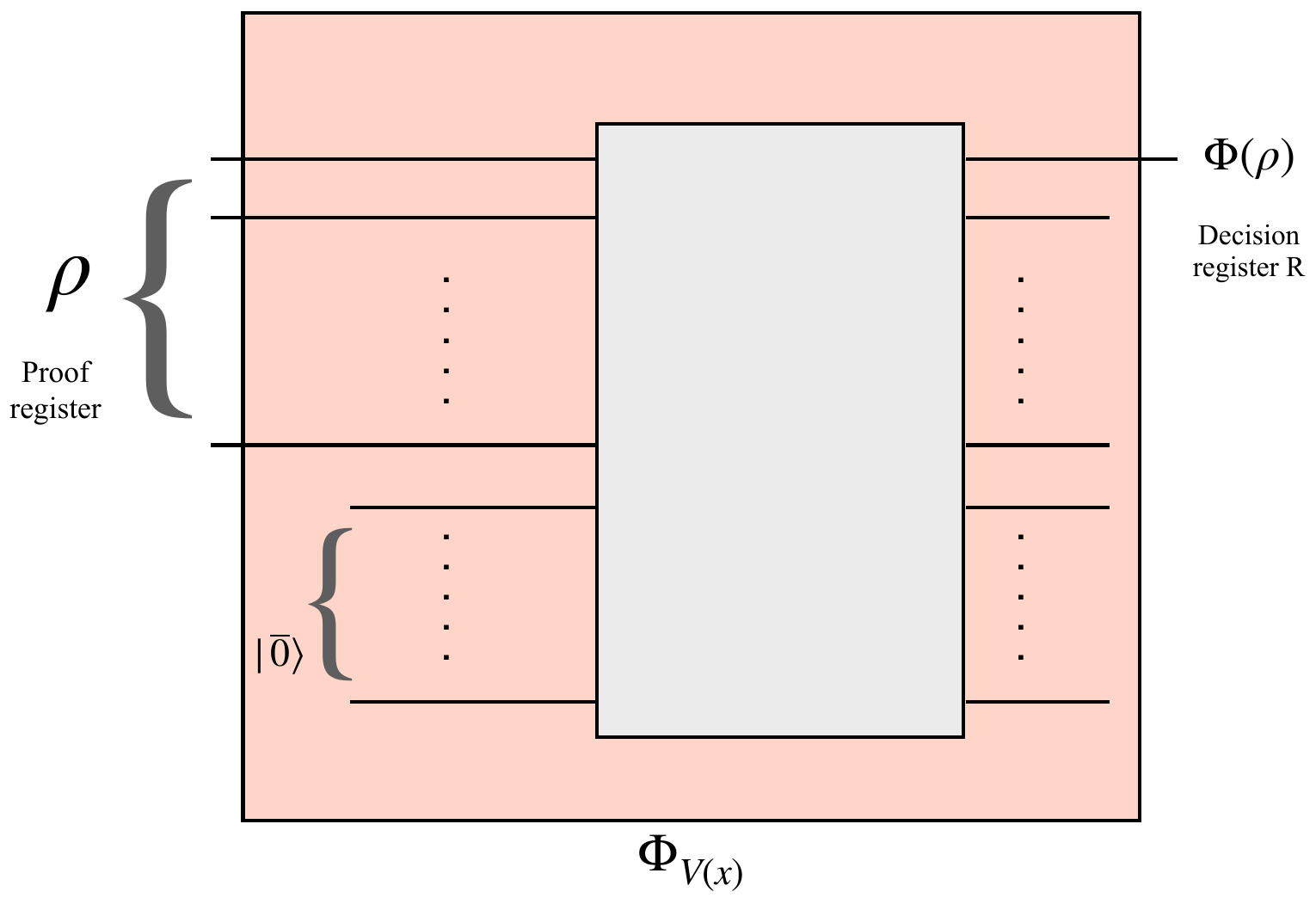}
    \caption{Verification circuit $V(x)$ that implements $\Phi_{V(x)}$.}
    \label{fig:QMA}
\end{figure} for a unitary $U_x:(\C^{2})^{\otimes (p(|x|) + l(|x|))} \to (\C^{2})^{\otimes (p(|x|) + l(|x|))}$, where the proof state is in $ (\C^{2})^{\otimes p(|x|)}$ and the ancillary register is initialized as $\ket{0}^{\otimes l(|x|)} \in (\C^{2})^{\otimes l(|x|)}$, such that the following hold: 
    \item For all $x \in Y$, there exists a density $\rho \in D\big((\C^2)^{\otimes p(|x|)}\big)$ such that
    \begin{equation*}
        \Tr\Big(U_x\big(\rho\otimes \ketbra{0}{0}^{\otimes l(\abs{x})} \big)U_x^\dagger (\ketbra{1}{1}\otimes I)   \Big)\geq c(\abs{x}).
    \end{equation*}
    \item For all $x \in N$ and $\rho \in D\big((\C^2)^{\otimes p(|x|)}\big)$
        \begin{equation*}
        \Tr\Big(U_x\big(\rho\otimes \ketbra{0}{0}^{\otimes l(\abs{x})} \big)U_x^\dagger (\ketbra{1}{1}\otimes I)   \Big)\leq s(\abs{x}).
    \end{equation*}
\end{itemize}

\end{definition}

For any $c, s : \mbb{N} \to (0, 1)_{\text{exp}}$, with $c(n)-s(n)=O(1/\text{poly}(n))$, we have that $\textsf{QMA}_{c, s} = \textsf{QMA}_{\frac{3}{4}, \frac{1}{4}}$ \cite{MW05}.

\begin{remark}\label{rem:QMA-Isometry}
    Given unitary $U_x$, as in \cref{def: class-qma}, we let $W_x:(\C^2)^{\otimes p(\abs{x})}\to (\C^{2})^{\otimes (p(|x|) + l(|x|))}$ denote the isometry sending $\ket{\psi}\mapsto U_x(\ket{\psi}\otimes\ket{0}^{\otimes l(\abs{x})})$. Note then that 
\begin{equation*}
    \Tr\Big(U_x\big(\rho\otimes \ketbra{0}{0}^{\otimes l(\abs{x})} \big)U^\dagger (\ketbra{1}{1}\otimes I)   \Big) = \Tr\Big(W_x\rho W_x^\dagger (\ketbra{1}{1}\otimes I)   \Big)= \Tr\Big( W_x^\dagger (\ketbra{1}{1}\otimes I) W_x\rho    \Big).
\end{equation*}
\end{remark}

\textsf{R} is the class of all problems that can be decided with a
Turing machine; $P \in \textsf{R}$ is called \textbf{decidable}, and $P \notin \textsf{R}$ is called \textbf{undecidable}. We define the class \textsf{RE} as follows. 

\begin{definition}[RE] \label{def: class-RE}
    \textsf{RE} is the class of all problems $P$ such that there exists a Turing machine $M$ where if $x \in Y(P), ~M$ halts on input $x$.
\end{definition}

\textsf{MIP} is the class of languages that can be decided by a multi-prover interactive proof system. $\textsf{MIP}(2, 1)$ corresponds to a proof system with two provers and a single round of interaction and is equivalent to families of nonlocal games $F = \{\mc{G}_x\}_x$ where the all-powerful provers attempt to convince the verifier that the word $x \in \{0, 1\}^{\ast}$ belongs to a language $L$. In a two-player nonlocal game $\mc{G}$, Alice and Bob (two players) play against a verifier. During the game, the verifier samples a question pair $(q, q')$ from a finite set $Q \times Q$ according to a probability distribution $\pi$, sends $q$ to Alice, $q'$ to Bob, and they respectively respond with answers $a, a' \in A$, another finite set denoting the answers. Alice and Bob win this round of the game if $V(q, q'; a, a) = 1$ and lose otherwise. Here, $V: Q \times Q' \times A \times A' \to \{0, 1\}$ is a function called the predicate. Importantly, while the setup of the game $\mc{G} = (Q, A, \pi, V)$ is known to everyone, Alice and Bob cannot communicate during the game and thus have no knowledge of each other's question and answer. The goal is to come up with a strategy which maximises their chance of winning. 

The power of $\textsf{MIP}(2,1)$ is captured by the \textit{gapped promise problem} for the family of games $F$: for fixed constants $0 < s < c \leq 1,$ given a game $\mc{G}_x \in F$ decide whether the value of $\mc{G}_x$ is at least $c$, or less than $s$, given the promise that one of the two must hold. This constitutes what is known as the $(c, s)$-gap problem for $F$ with completeness and soundness parameters $c$ and $s$ respectively. Since the provers in \textsf{MIP} cannot communicate, it is natural to ask what happens if they are allowed to share entanglement which gives rise to the complexity class $\textsf{MIP}^{\ast}$ with entangled provers. $\textsf{MIP}^{\ast}$ is captured by the gapped promise problem defined above where the value of a game is replaced by its quantum value \cite{CHTW04}. 

The complexity of $\textsf{MIP}^{\ast},$ the quantum counterpart of $\textsf{MIP}$ turned out to be much harder to determine and increasingly larger until it was shown to be equal to \textsf{RE}, the class of recursively enumerable languages (see \cref{sec: cmplx-th}) in \cite{JNV+21}. More precisely, for any $0 < s < 1$ the so-called \textit{synchronous}\footnote{In a synchronous game, if the players receive the same question, they must produce the same answer in any perfect strategy games.} games have been constructed for which the $(1, s)$-gap$^{\ast}$ problem decides the Halting problem. This result indicates that, in general, approximating the quantum value of synchronous games within any additive constant is undecidable. Immediately prior to this, \cite{Slo19} showed an important result that whether or not a synchronous game has a perfect quantum strategy is undecidable.

\subsection{Quantum independence number} \label{sec: q-ind-num}

In this section, we introduce the classical and quantum independence numbers of a graph, and briefly discuss the computational complexity of computing the latter.

Let $G=(V,E)$ be a graph, where $V$ is a finite set of vertices and $E\subseteq V\times V$ is the set of edges. Throughout this paper, we assume all graphs are \textbf{undirected} and \textbf{simple}, meaning that $(v,v)\neq E$ for all $v\in V$, and $(v,w)\in E$ implies $(w,v)\in E$. A subset $S\subseteq V$ is an \textbf{independent set} if no two distinct vertices in $S$ are adjacent, i.e., $(v,w)\notin E$ for all $v,w\in S$. The (classical) independence number of $G$ is
\begin{equation*}
    \alpha(G):=\max\{\abs{S}: S\subseteq V \text{ is an independent set}\}.
\end{equation*}

Given a graph $G=(V,E)$ and a positive integer $t$, the \textbf{$t$-independent set game }$\mathrm{IS}(G,t)$ of $G$ is a two-player (commonly called Alice and Bob) nonlocal game in which Alice and Bob attempt to convince a referee that $G$ contains an independent set of $t$ vertices. In $\mathrm{IS}(G,t)$, the referee samples $i,j\in [t]$ uniformly and sends $i$ to Alice and $j$ to Bob. Alice and Bob then respond with vertices $v,w\in V$, respectively. The players win if and only if:
\begin{itemize}
    \item if $i= j$, then $v=w$;
    \item if $i\neq j$, then $(v,w)\notin E$.
\end{itemize}
Alice and Bob may agree on a strategy beforehand, but are not allowed to communicate during the game. A strategy is said to be \textbf{perfect} if it wins this game with certainty. 

A \textbf{deterministic classical strategy} consists of two functions $f_A,f_B:[t]\rightarrow V$. On questions $i$ and $j$, Alice outputs $f_A(i)$ and Bob outputs $f_B(j)$. It is not hard to see that $\mathrm{IS}(G,t)$ has a perfect deterministic classical strategy if and only if $\alpha(G)\geq t$. This gives an alternative formulation of $\alpha(G)$: it is the maximal $t$ such that $\mathrm{IS}(G,t)$ has a perfect classical strategy.

A \textbf{quantum strategy} allows Alice and Bob to share entanglement, often leading to a higher winning probability than classically possible. Such a strategy $S$ consists of a quantum state $\ket{\psi}\in \mc{H}_A\otimes\mc{H}_B$, where $\mc{H}_A$ and $\mc{H}_B$ are finite-dimensional Hilbert spaces, and POVMs $\{P^i_v:v\in V\}\subseteq B(\mc{H}_A)$ and $\{Q^j_w:w\in V\}\subseteq B(\mc{H}_B)$ for all $i,j\in [t]$. When the players use this strategy, the probability of answering $v,w$ on questions $i,j$ is
\begin{equation*}
    p_S(v,w|i,j)=\bra{\psi}P^i_v\otimes Q^j_w\ket{\psi}.
\end{equation*}
The \textbf{quantum independence number}  of $G$ is
\begin{equation*}
 \alpha_q(G)=\max \{t: \mathrm{IS}(G,t) \text{ has a perfect quantum strategy}  \}.
\end{equation*}
Quantum independent set games belong to a broader class called \textbf{symmetric synchronous games}.
In a synchronous game, if the players receive the same question, they must produce the same answer in any perfect strategy. A synchronous game is symmetric if exchanging the roles of Alice and Bob does not change the winning rules. Perfect strategies for synchronous games are well-understood~\cite{HMPS17}. In particular, we know that $\mathrm{IS}(G,t)$ has a perfect quantum strategy if and only if there exist finite-dimensional PVMs $\{P^i_v:v\in V\},i\in [t]$ such that $P^i_vP^j_w=0$ whenever answering $v$ and $w$ for questions $i$ and $j$ loses the game. This gives an alternative formulation for the quantum independence number:
\begin{align*}
    \alpha_q(G)&= \max \{ t: \text{there exist finite-dimensional PVMs}~ \{P^i_v:v\in V\}, i\in [t]\\
    & \hspace{2cm}\text{such that for }i\neq j, P^i_vP^j_w=0 \text{ whenever } v=w \text{ or } (v,w)\in E \}.
\end{align*}
The $\textsf{MIP}^{\ast} = \textsf{RE}$ result~\cite{JNV+21} implies that the decision problem
 \begin{quote}
     Given a synchronous game $\mc{G}$, does it have a perfect quantum strategy?
 \end{quote}
is $\textsf{RE}$-complete. We remark that this problem was first shown to be undecidable by Slofstra in \cite{Slo19}. The $\textsf{MIP}^{\ast} = \textsf{RE}$ result strengthens it to $\textsf{RE}$-complete. See also \cite{PS25} for a detailed discussion.

Given any synchronous game $\mc{G}$ with question set $Q$ and answer set $A$, we can construct a graph $G=(V,E)$ with $V=Q\times A$, and $((q,a),(q',a'))\in E$ if and only if answering $a,a'$ on $q,q'$ loses the game $\mc{G}$. Let $t=\abs{Q}$. Then the independent set game $\mathrm{IS}(G,t)$ has a perfect quantum strategy if and only if $\mc{G}$ has a perfect quantum strategy. Since this mapping from synchronous games $\mc{G}$ to independent set games $\mathsf{IS}(G,t)$ is clearly computable, as an immediate consequence: 
 \begin{theorem}\label{thm:Gt}
     The decision problem
 \begin{quote}
     Given a graph $G$ and a positive integer $t$, is $\alpha_q(G)\geq t$?
 \end{quote}
 is $\textsf{RE}$-hard.
 \end{theorem}

\section{Computing the quantum capacity is \textsf{QMA}-hard} \label{sec: QMA-hard}

In this section we show that computing the quantum capacity of a quantum channel is \textsf{QMA}-hard (\cref{cor:QMA-hard}). To do so, we introduce the following construction, a \textbf{symmetric direct sum} of quantum channels.

\begin{definition}[Symmetric direct sum] \label{def:symm-dir-sum}
    Let $U : \mc{H}_A \to \mc{H}_{A'}$ be an isometry, where $\mc{H}_{A'}=\C^{2^n}$ with the standard orthonormal basis $\{\ket{z}:z\in\{0,1\}^n\}$. Let $\Phi_0 : B(\mc{H}_B) \to B(\mc{H}_{B'})$ and $\Phi_1 : B(\mc{H}_B) \to B(\mc{H}_{B'})$ be CPTP maps. Define the \emph{symmetric direct sum} of $\Phi_0$ and $\Phi_1$ as the channel $\Phi_0 \oplus_U \Phi_1 : B(\mc{H}_A\otimes\mc{H}_B)\to B(\mc{H}_{A'}\otimes\mc{H}_{B'})$ given by
    \begin{equation}
        \Phi_0 \oplus_{U} \Phi_1 (\rho) \coloneqq \sum \limits_{z \in \{0, 1\}^n} (\ketbra{z} U \otimes I_{B'}) (\text{id}_A \otimes \Phi_{z_1}) (\rho) (U^{\dag} \ketbra{z} \otimes I_{B'}).
    \end{equation}
    Here, $z_1$ is the first bit of $z$.    
\end{definition}

The following lemma shows that the complement of the symmetric direct sum of $\Phi_0$ and $\Phi_1$ is equal to the symmetric direct sum of the complementary channels $\Phi_0^c$ and $\Phi_1^c.$

\begin{lemma} \label{lem:comp-direct-sum}
    $(\Phi_0 \oplus_U \Phi_1)^c = \Phi^c_0 \oplus_U \Phi^c_1$,
    where all symbols are as in \cref{def:symm-dir-sum}.
\end{lemma}

\begin{proof}
    Let $V_0 : \mc{H}_B \to \mc{H}_{B'E}$ and $V_1 : \mc{H}_B \to \mc{H}_{B'E}$ be the Stinespring isometries of $\Phi_0$ and $\Phi_1.$ Then, we claim that $V = \sum \limits_z \ket{zz}_{A'A'_E}\bra{z} U \otimes V_{z_1}$ is the Stinespring isometry of $\Phi_0 \oplus_U \Phi_1$ where $A'_E$ is a copy of register $A'$ in the environment of $\Phi_0 \oplus_U \Phi_1 (\rho)$.
    \begin{align*}
        \Tr_{A'_EE}(V\rho V^\dag)&= \Tr_{A'_EE}\left(\sum \limits_z \ket{zz}\bra{z} U \otimes V_{z_1}\right) \rho \left(\sum \limits_{z'} U^{\dag} \ket{z'}\bra{z'z'} \otimes V^{\dag}_{z'_1}\right) \\
        &= \sum \limits_z \sum \limits_{z'} \Tr_{A'_EE} \big(\ket{zz}\bra{z} U \otimes V_{z_1}\big) \rho \big(U^\dag\ket{z'}\bra{z'z'} \otimes V^{\dag}_{z'_1}\big) \\
        &= \sum \limits_z \sum \limits_{z'} \Tr_{E} \big(\ket{z}\bra{z} U \otimes V_{z_1}\big) \rho \big(U^\dag\ket{z'}\bra{z'} \otimes V^{\dag}_{z'_1}\big)\braket{z'}{z} \\
        &= \sum \limits_z \Tr_{E} \big(\ket{z}\bra{z} U \otimes V_{z_1}\big) \rho \left(U^\dag\ket{z}\bra{z} \otimes V^{\dag}_{z_1}\right) \\
        &= \sum_z (\ketbra{z} U \otimes I_{B'}) (\text{id}_A \otimes \Phi_{z_1}) (\rho) (U^{\dag} \ketbra{z} \otimes I_{B'}) \\
        &= \Phi_0 \oplus_U \Phi_1 (\rho).
    \end{align*}

    We have that $\Phi_0^c (\rho) = \Tr_{B'} (V_0 \rho V^{\dag}_0)$ and $\Phi_1^c (\rho) = \Tr_{B'} (V_1 \rho V^{\dag}_1)$. Therefore,
    \begin{align*}
        \Phi_0^c \oplus_U \Phi_1^c (\rho) &= \Tr_{B'} (V_0 \rho V^{\dag}_0) \oplus_U \Tr_{B'} (V_1 \rho V^{\dag}_1) \\
        &=\sum \limits_z (\ketbra{z} U \otimes I_E) (\text{id}_A \otimes \Phi_{z_1}^c) (\rho) (U^{\dag} \ketbra{z} \otimes I_E).
    \end{align*}

    We now consider the complement of the symmetric direct sum.
    \begin{align*}
        (\Phi_0 \oplus_U \Phi_1)^c (\rho) &= \Tr_{A'B'} \left(\sum \limits_z \ket{zz}\bra{z} U \otimes V_{z_1}\right) \rho \left(\sum \limits_{z'} U^{\dag} \ket{z'}\bra{z'z'} \otimes V^{\dag}_{z'_1}\right) \\
        &=\sum \limits_z \sum \limits_{z'} \Tr_{A'B'} \left(\ket{zz}\bra{z} U \otimes V_{z_1}\right) \rho \left(U^\dag\ket{z'}\bra{z'z'} \otimes V^{\dag}_{z'_1}\right) \\
        &=\sum \limits_z \Tr_{B'} (\ketbra{z} U \otimes V_{z_1}) \rho (U^{\dag} \ketbra{z} \otimes V_{z_1}) \\
        &= \Phi^c_0 \oplus_U \Phi^c_1 (\rho).
    \end{align*}
    This completes the proof.
\end{proof}

In the next lemma we show that the channel $\Phi_0 \oplus_U \Phi_1$ is degradable if $\Phi_0$ and $\Phi_1$ are degradable.

\begin{lemma} \label{lem:degr-dir-sum}
    Let $\Phi_0, \Phi_1$ be degradable channels and let $\Phi_0 \oplus_U \Phi_1$ be as constructed in \cref{def:symm-dir-sum}. Then, $\Phi_0 \oplus_U \Phi_1$ is a degradable channel.
\end{lemma}

\begin{proof}
    Let $\mc{D}_0$ and $\mc{D}_1$ be degrading maps for $\Phi_0$ and $\Phi_1$ such that $\mc{D}_0 \circ \Phi_0 = \Phi_0^c$ and $\mc{D}_1 \circ \Phi_1 = \Phi^c_1.$ Define the channel $\mc{D}:B(\mc{H}_{A'}\otimes\mc{H}_{B'})\rightarrow B(\mc{H}_{A'_E}\otimes\mc{H}_{E})$ as
    \begin{align*}
        \mc{D}(\rho)=\sum_{z\in\{0,1\}^n}(\ketbra{z}\otimes I_E)(\mathrm{id}_{A'}\otimes\mc{D}_{z_1})(\rho)(\ketbra{z}\otimes I_E).
    \end{align*}
    We claim that $\mc{D}$ is a degrading map for $\Phi_0\oplus_U\Phi_1$. 
    
    Now we will show that $\mc{D} \circ (\Phi_0 \oplus_U \Phi_1)~=~(\Phi_0 \oplus_U \Phi_1)^c.$
    \begin{align*}
        &\mc{D} \circ (\Phi_0 \oplus_U \Phi_1) (\rho) \\
        =&\mc{D} \left(\sum \limits_{z \in \{0, 1\}^n} (\ketbra{z} U \otimes I_{B'}) (\text{id}_A \otimes \Phi_{z_1}) (\rho) (U^{\dag} \ketbra{z} \otimes I_{B'})\right) \\
        =&\sum_{z,z'\in\{0,1\}^n}(\ketbra{z}\otimes I_E)(\mathrm{id}_{A'}\otimes\mc{D}_{z_1}) \left( (\ketbra{z'} U \otimes I_{B'}) (\text{id}_A \otimes \Phi_{z'_1}) (\rho) (U^{\dag} \ketbra{z'} \otimes I_{B'}) \right) (\ketbra{z}\otimes I_E) \\
        =&\sum_{z\in\{0,1\}^n}(\ketbra{z} U \otimes I_E )(\text{id}_A \otimes \mc{D}_{z_1} \circ \Phi_{z_1}) (\rho) (U^{\dag} \ketbra{z}\otimes I_E) \\
        =&\sum_{z\in\{0,1\}^n}(\ketbra{z} U \otimes I_E )(\text{id}_A \otimes \Phi_{z_1}^c) (\rho) (U^{\dag} \ketbra{z}\otimes I_E) \\
        =&(\Phi_0 \oplus_U \Phi_1)^c (\rho),
    \end{align*}
    where the last equality follows from \cref{lem:comp-direct-sum}. This completes the proof.
\end{proof}

An immediate corollary results from the above lemma just by the property that if a channel is degradable, then its quantum capacity is equal to its one-shot quantum capacity \cite{CRS08}.

\begin{corollary} \label{cor:Q=Q1-symm-dir-sum}
    Let $\Phi_0, \Phi_1$ be degradable channels and let $\Phi_0 \oplus_U \Phi_1$ be as constructed in \cref{def:symm-dir-sum}. Then, $Q(\Phi_0 \oplus_U \Phi_1) = Q^{(1)}(\Phi_0 \oplus_U \Phi_1)$.
\end{corollary}

We now introduce the specific channels we choose as $\Phi_0$ and $\Phi_1$ to establish our \textsf{QMA}-hardness result, \cref{thm:main-QMA-hardness}.

\begin{definition} \label{def:Q_}
    Let $\Phi$ be a quantum channel. We define $Q^{(1)}_-(\Phi) \coloneqq \min\limits_\rho I_c(\Phi,\rho)$.
\end{definition}

\begin{remark} \label{rem:amp-damp}
    Consider the two extreme cases of the \textit{qubit amplitude damping channel}. This channel denoted by $\mc{A}_{\eta}$, is specified by two Kraus operators $A_0 = \begin{pmatrix} 
        1 & 0 \\
        0 & \sqrt{1 - \eta}
      \end{pmatrix}$ and $A_1
    =  \begin{pmatrix}
          0 & \sqrt{\eta} \\
          0 & 0
      \end{pmatrix}$ such that $A_{\eta}(\rho) = A_0 \rho A_0^{\dag} + A_1 \rho A_1^{\dag}.$ Here, $\eta$ is interpreted as the probability of transmission. This channel physically represents the loss of an excitation (say, a photon) from a quantum system to a noisy environment. This channel is degradable for $\eta \in \left[0, \frac{1}{2} \right]$ and anti-degradable for $\eta \in \left[\frac{1}{2}, 1\right].$ Notably, it is both degradable and anti-degradable at $\eta = 1/2.$ At $\eta =0$, $\mc{A}_0 = \text{id}_2$ and hence $Q(\mc{A}_0) = Q^{(1)}(\mc{A}_0) = \log 2 = 1$. Since the quantum capacity is zero for any anti-degradable channel we have that $Q(\mc{A}_{1/2}) = Q^{(1)}(\mc{A}_{1/2}) = 0$. Further, since $I_c(\mc{A}_{1/2},\rho)=0$ for all $\rho$, $Q^{(1)}_-(\mc{A}_{1/2})=0$, as defined in \cref{def:Q_}, as well. We choose $\Phi_1 = \mc{A}_0$ and $\Phi_0 = \mc{A}_{1/2}$.
\end{remark}

We now state the following known result for the von-Neumann entropy of subnormalised states. It is important for establishing a result about the coherent information of the symmetric direct sum of channels, \cref{thm:coh-inf-symm-dir-sum}, hence we include a proof.

\begin{lemma} \label{lem: subnorm-vN-ent}
    For every $q \geq 0$ and $\omega \in D(\mc{H})$, $S(q \omega) = q S(\omega) - q \log q$.    
\end{lemma}

\begin{proof}
Since $\Tr(\omega) = 1$, we obtain that 
    \begin{align*}
        S(q \omega) &= - \Tr \big(q \omega \log (q \omega)\big) \\
        &= - \Tr (q \omega \log \omega) - \Tr (q \omega \log q) \\
        &= - q \Tr(\omega \log \omega) -q \log q \Tr(\omega) \\
        &= q S(\omega) - q \log q.
    \end{align*}
This completes the proof.
\end{proof}

The next lemma gives us the von-Neumann entropy of classical-quantum states, another known result important for establishing \cref{thm:coh-inf-symm-dir-sum}.

\begin{lemma} [\cite{YHD08}] \label{lem:vN-ent-cq}
    Let $\rho = \sum \limits_z \ketbra{z} \otimes \rho_z$ be a classical-quantum state. Then
    \begin{equation}
        S(\rho) = \sum \limits_z S(\rho_z).
    \end{equation}
\end{lemma}

The following theorem gives the coherent information of the symmetric direct sum of quantum channels.

\begin{theorem} \label{thm:coh-inf-symm-dir-sum}
    Let $\Phi_0 \oplus_U \Phi_1 (\rho)$ be as in \cref{def:symm-dir-sum}. Let $Q^{(1)}(\Phi_1) \geq Q^{(1)}(\Phi_0).$ Let $p = \max \limits_{\sigma} \Tr\left( (U^{\dag} (\ketbra{1} \otimes I) U) \sigma \right)$ where the maximisation is over quantum states $\sigma$. Then, writing $Q^{(1)}_-(\Phi_0)$ as in \cref{def:Q_}, we have that
    \begin{equation}
        p Q^{(1)}(\Phi_1) + (1-p)Q^{(1)}(\Phi_0)\geq Q^{(1)} (\Phi_0 \oplus_U \Phi_1)\geq p Q^{(1)}(\Phi_1) + (1-p)Q^{(1)}_-(\Phi_0). \label{eq:capacity-symm-dir-sum}
    \end{equation}
\end{theorem}

\begin{proof}
    Define $p_z=\Tr(\bra{z}U\otimes I_{B})\rho(U^\dag\ket{z}\otimes I_{B})$ and $\rho_z=\frac{1}{p_z}(\bra{z}U\otimes I_{B})\rho(U^\dag\ket{z}\otimes I_{B})$. Then, $(\Phi_0\oplus_U\Phi_1)(\rho)=\sum \limits_z p_z\ketbra{z}\otimes\Phi_{z_1}(\rho_z)$. 
    Now \begin{align}
        S(\Phi_0\oplus_U\Phi_1 (\rho)) &= S\left(\sum \limits_z p_z \ketbra{z} \otimes \Phi_{z_1} (\rho_z)\right) \nonumber \\
        &= \sum \limits_z S \big(p_z \Phi_{z_1} (\rho_z)\big) \nonumber \\
        &= \sum \limits_z \Big(p_z S(\Phi_{z_1}(\rho_z)) - p_z \log p_z \Big), \label{eq:entropy-coh}
    \end{align} 
    where the second equality is by \cref{lem:vN-ent-cq} and the last equality is by \cref{lem: subnorm-vN-ent}.
    Then,
    \begin{align*}
        S((\Phi_0\oplus_U\Phi_1)^c (\rho)) &= S((\Phi^c_0 \oplus_U \Phi^c_1) (\rho)) \\
        &= \sum \limits_z \Big(p_z S(\Phi^c_{z_1}(\rho_z)) - p_z \log p_z \Big),
    \end{align*}
    where the first equality is by \cref{lem:comp-direct-sum} and the last equality is by \cref{eq:entropy-coh}.
    Now,
    \begin{align*}
        I_c (\Phi_0 \oplus_U \Phi_1, \rho) &= S(\Phi_0 \oplus_U \Phi_1 (\rho)) - S(\Phi_0 \oplus_U \Phi_1)^c (\rho) \\
        &= \sum \limits_z \big(p_z S(\Phi_{z_1}(\rho_z)) - p_z \log p_z \big) - \sum \limits_z \big(p_z S(\Phi^c_{z_1}(\rho_z)) - p_z \log p_z \big) \\
        &= \sum \limits_z p_z I_c(\Phi_{z_1}, \rho_z). 
    \end{align*}
 Then,
    \begin{align*}
        Q^{(1)}(\Phi_0 \oplus_U \Phi_1) &= \max \limits_{\rho}\sum \limits_z p_z I_c(\Phi_{z_1}, \rho_z) \\
        &\leq \max \limits_{\rho} \sum \limits_{z} p_z Q^{(1)}(\Phi_{z_1}) \\
        &= p Q^{(1)}(\Phi_1) + (1-p)Q^{(1)}(\Phi_0).
    \end{align*}
   The last equality follows from the assumption that $Q^{(1)}(\Phi_1)\geq Q^{(0)}(\Phi_0)$. For the lower bound, let $\sigma$ be a state such that $p=\Tr\left( (U^{\dag} (\ketbra{1} \otimes I) U) \sigma \right)$ and $\rho$ be a state such that $I_c(\Phi_1,\rho)=Q^{(1)}(\Phi_1)$. Then,
    \begin{align*}
        Q^{(1)}(\Phi_0 \oplus_U \Phi_1) &\geq I_c(\Phi_0 \oplus_U \Phi_1,\sigma\otimes\rho)\\
        &= \sum \limits_z p_z I_c(\Phi_{z_1}, \rho) \\
        &= p I_c(\Phi_1,\rho) + (1-p)I_c(\Phi_0,\rho)\\
        &\geq p Q^{(1)}(\Phi_1) + (1-p)Q^{(1)}_-(\Phi_0).
    \end{align*}
\end{proof}

The following corollary immediately holds.

\begin{corollary} \label{cor:special-channels-dir-sum}
    Let $\Phi_0 = \mc{A}_{1/2}$ and $\Phi_1 = \mc{A}_0$ as in \cref{rem:amp-damp}. Let $p$ be as in \cref{thm:coh-inf-symm-dir-sum}. Then, $Q^{(1)}(\Phi_0\oplus_U\Phi_1)=p$.
\end{corollary}

\begin{figure}
    \centering
    \includegraphics[width=1\textwidth]{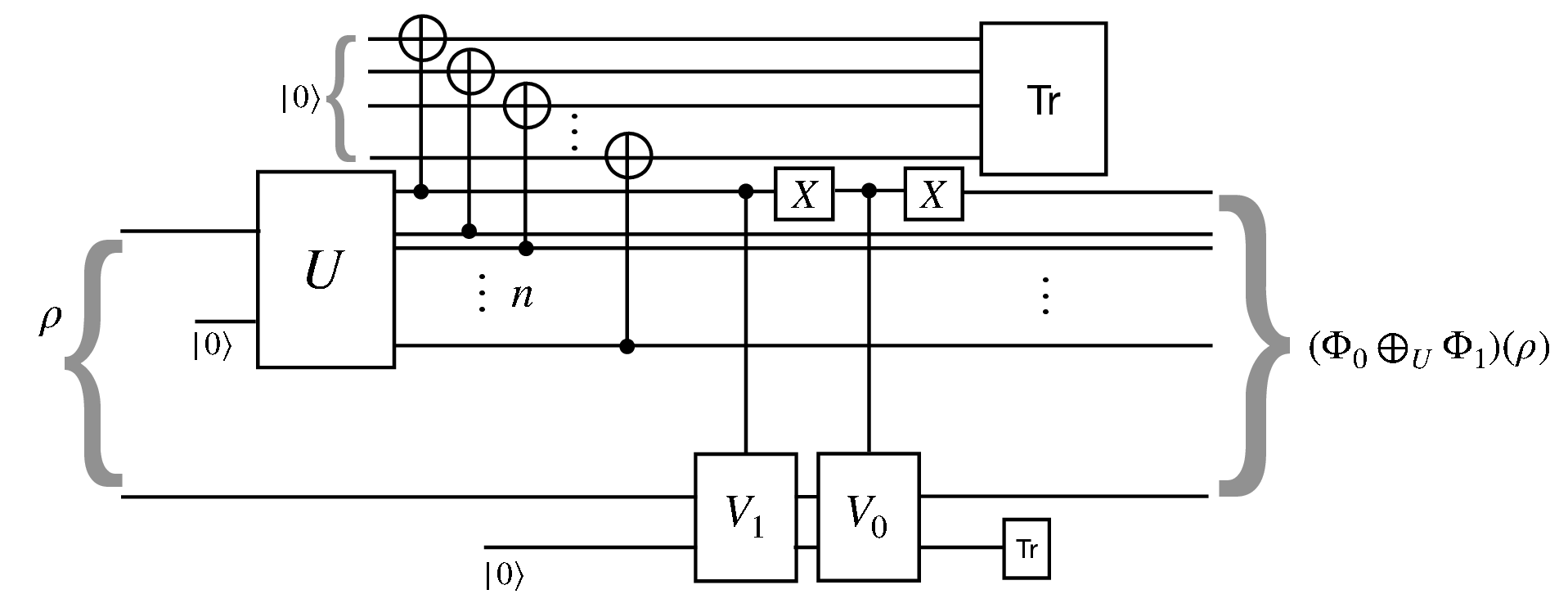}
    \caption{The symmetric direct sum of channels succinctly presented by a circuit.}
    \label{fig:circuit}
\end{figure}

In the remark below, we show that there exists a circuit implementing the symmetric direct sum of channels as illustrated in \cref{fig:circuit}.

\begin{remark} \label{rem:cir-dir-sum}
    Let $C_0$, $C_1$ be circuits representing the channel actions of $\Phi_0$ and $\Phi_1.$ From the Stinespring representation of $\Phi_0(\rho) = (\text{id} \otimes \Tr)(V_0 (\rho \otimes \ketbra{0}) V^{\dag}_0), ~\Phi_1(\rho) = (\text{id} \otimes \Tr)(V_1 (\rho \otimes \ketbra{0}) V^{\dag}_1)$ we can, without loss of generality, let $C_1$ and $C_0$ be represented by the action of the unitary $V_1$ and $V_0$ respectively, where $\ketbra{0}$ denotes an auxiliary input system. Let $C$ be the circuit corresponding to the unitary $U$. We now construct the circuit for the channel $\Phi_0 \oplus_U \Phi_1$ using the diagram in \cref{fig:circuit}. 
\end{remark}

The following theorem establishes a reduction from every promise problem in \textsf{QMA} to the problem of deciding the quantum capacity of symmetric direct sum channels.

\begin{theorem} \label{thm:main-QMA-hardness}
    For every promise problem $L \in \textsf{QMA}_{\frac{3}{4}, \frac{1}{4}}$, there is an efficient mapping from instances $x$ to circuit descriptions implementing channels $\Phi_x$ such that if $x\in Y$ then $Q(\Phi_x)\geq \frac{3}{4}$, and if $x\in N$ then $Q(\Phi_x) \leq \frac{1}{4}$.
\end{theorem}

\begin{proof}
    Let $L=(Y,N) \in \textsf{QMA}_{\frac{3}{4},\frac{1}{4}}$ as in \cref{def: class-qma},
    and let $V(x)$ be the corresponding verification circuit for each instance $x$.  The circuit $V(x)$ acts by a unitary $U_x:\C^{2^{n_1}}\otimes \C^{2^{n_2}} \to \C^{2^{n_1+n_2}}$ on the proof state in $ \C^{2^{n_1}}$ with an ancilla $\ket{0}^{\otimes n_2} \in \C^{2^{n_2}}$, and then measures the first qubit in the computational basis. Here $n_1=p(\abs{x})$ and $n_2=l(\abs{x})$ as defined in \Cref{def: class-qma}. Define the isometry $W_x:\C^{2^{n_1}}\to \C^{2^{n_1+n_2}}$ by sending $\ket{\psi}\mapsto U_x(\ket{\psi}\otimes\ket{0}^{\otimes n_2})$, as discussed in \cref{rem:QMA-Isometry}. Then the probability of accepting $\rho$ is given by $\Tr(W_x^\dagger (\ketbra{1}{1}\otimes I)W_x\rho)$.
    
    Let $\Phi_x = \Phi_0\oplus_{W_x}\Phi_1$ be the symmetric direct sum where $\Phi_0 = \mc{A}_{1/2}$ and $\Phi_1 = \mc{A}_0$ are the extremes of the qubit amplitude damping channel $\mc{A}_{\eta}$ as in \cref{rem:amp-damp}. By \cref{rem:cir-dir-sum}, the circuit description for $\Phi_x$ can be computed in polynomial time from the classical description of the circuit $V(x)$. Then by \cref{cor:special-channels-dir-sum}, we have $Q(\Phi_x)= \max_\rho \Tr(W_x^\dagger (\ketbra{1}{1}\otimes I)W_x\rho).$
    
    If $x \in Y$ then there exists a quantum state $\rho$ such that $\Tr(W_x^\dag(\ketbra{1}{1}\otimes I)W_x\rho) \geq \frac{3}{4}$. Thus the quantum capacity $$Q(\Phi_x)=Q^{(1)}(\Phi_0\oplus_{W_x}\Phi_1)\geq\Tr(W_x^\dag(\ketbra{1}{1}\otimes I)W_x\rho)\geq\frac{3}{4}.$$
    If $x \in N$ then for every quantum state $\rho$, $\Tr(W_x^\dag(\ketbra{1}{1}\otimes I)W_x\rho) \leq \frac{1}{4}$. Then, as above, the quantum capacity $$Q(\Phi_x)=Q^{(1)}(\Phi_0\oplus_{W_x}\Phi_1)\leq\frac{1}{4}.$$
\end{proof}

This results in the immediate corollary stated as follows.

\begin{corollary} \label{cor:QMA-hard}
    Given a quantum channel $\Phi_{\mc{C}} : B(\mc{H}) \to B(\mc{K})$, presented succinctly by a quantum circuit $\mc{C}$ which implements $\Phi_{\mc{C}}$, it is $\mathrm{QMA}$-hard to decide if $Q(\Phi_{\mc{C}}) \geq \frac{3}{4}$ or $Q(\Phi_{\mc{C}}) \leq \frac{1}{4}$.
\end{corollary}

\section{A special case of the entanglement-assisted zero-error classical capacity of quantum channels is uncomputable} \label{sec: RE-hard}

In this section, we establish the following relations between the capacity $C^{(1)}_{0,\PME}$ of quantum channels and the quantum independence number of graphs.

\begin{theorem}\label{theorem:main=}
There is a computable mapping from (simple and undirected) graphs $G$ to classical-to-quantum channels $\Phi_G$ such that $C^{(1)}_{0,\PME}(\Phi_G)= \log \left(\alpha_q(G)\right)$.
\end{theorem}

As an immediate consequence of \Cref{thm:Gt} that determining whether $\alpha_q(G)\geq t$ is \textsf{RE}-hard, we have the following undecidability result for $C^{(1)}_{0,\PME}$.
\begin{corollary}\label{theorem:RE}
    The decision problem
    \begin{quote}
        Given a positive integer $t$ and a classical-to-quantum channel $\Phi$, is $C^{(1)}_{0,\PME}(\Phi)\geq \log t$?
    \end{quote}
    is $\textsf{RE}$-hard.
\end{corollary}

Now we prove \cref{theorem:main=}.

\begin{proof}[Proof of \cref{theorem:main=}]
We first describe the mapping from graphs $G$ to c-q channels $\Phi_G$. Given any (simple and undirected) graph $G=(V,E)$, consider the matrix $A=(A_{vw})\in \C^{V\times V}$ defined by $A_{vw}=\begin{cases}
    1 & \text{ if } v= w \text{ or } (v,w)\in E\\
    0 & \text{ otherwise}
\end{cases}$. Let $\alpha:=\abs{\min \lambda(A)}$ be the absolute value of the minimal eigenvalue of $A$. Then $\alpha I+A$ is a positive semidefinite matrix with diagonal elements $1+\alpha$. We then apply Cholesky decomposition to find vectors $\{\ket{\psi_v}:v\in V\}\in \C^{V}$ such that
\begin{equation*}
    \braket{\psi_v}{\psi_w}=\begin{cases}
        1 & \text{ if } (v,w)\in E\\
        1+ \alpha & \text{ if } v=w\\
        0 & \text{ otherwise.}
    \end{cases}
\end{equation*}
Let 
\begin{equation}
   \ket{\tilde{\psi}_v}:=\frac{1}{\sqrt{1+\alpha}}\ket{\psi_v} \label{eq:tpsi}
\end{equation}
for all $v\in V$. Then $\{\ket{\tilde{\psi}_v}:v\in V  \}$ are unit vectors in $\C^V$ such that
\begin{equation*}
\braket{\tilde{\psi}_v}{\tilde{\psi}_w}=\begin{cases}
    \frac{1}{1+\alpha} & \text{ if } (v,w)\in E\\
    0 & \text{ if } v\neq w \text{ and }(v,w)\notin E.
\end{cases}
\end{equation*}
Define the c-q channel $\Phi_G:B(\C^V)\rightarrow B(\C^V)$ by sending
\begin{equation*}
    \ket{v}\bra{v}\mapsto \ket{\tilde{\psi}_v}\bra{\tilde{\psi}_v}
\end{equation*}
for all $v\in V$, where $\{\ket{v}:v\in V\}$ is the standard orthonormal basis for $\C^V$.

    We first prove that $C^{(1)}_{0,\PME}(\Phi_G)\geq  \log \left(\alpha_q(G)\right)$. Let $t:= \alpha_q(G)$. Then there are PVMs $\{P^i_v:v\in V\}\subseteq B(\mc{H}), i\in [t]$ for some finite-dimensional space $\mc{H}$  such that for $i\neq j$, $P^i_vP^j_w=0$ whenever $ v=w \text{ or } (v,w)\in E$. Let $\ket{\tau}\in \mc{H}\otimes\mc{H}$ be a maximally entangled state. For every $i\in [t]$, let $\mc{E}_i:B(\mc{H})\rightarrow B(\C^V)$ be the q-c channel sending
    \begin{equation*}
        \rho \rightarrow \sum_{v\in V} \Tr(P^i_v\rho)\ket{v}\bra{v},
    \end{equation*}
    and let 
    \begin{equation*}
        \rho_i= \left(\Phi_G\otimes \id_H  \right)\circ \left( \mc{E}_i\otimes \id_H \right)(\ket{\tau}\bra{\tau}).
    \end{equation*}
    Then for any $i\neq j$,
    \begin{align*}
        \Tr(\rho_i\rho_j)=\frac{1}{\dim(\mc{H})^2}\sum_{v,w\in V}\abs{\braket{\tilde{\psi}_v}{\tilde{\psi}_w}}^2 \Tr(P^i_vP^j_w).
    \end{align*}
When $v=w$ or $(v,w)\in E$, $Tr(P^i_vP^j_w)=0$; when $v\neq w$ and $(v,w)\notin E$, $\braket{\tilde{\psi}_v}{\tilde{\psi}_w}=0$. It follows that $\Tr(\rho_i\rho_j)=0$ for all $i\neq j$. Hence $\rho_1,\ldots,\rho_t$ are perfectly distinguishable. This implies $C^{(1)}_{0,\PME}(\Phi_G)\geq \log t$.

Now we prove $C^{(1)}_{0,\PME}(\Phi_G)\leq  \log \left(\alpha_q(G)\right)$. Let $t\in \N$ such that $C^{(1)}_{0,\PME}(\Phi_G)= \log t$.
Then there exists a finite-dimensional space $\mc{H}_R$, a finite-dimensional space $\mc{H}_A$ and a maximally entangled state $\ket{\psi}\in \mc{H}_A\otimes \mc{H}_R$, and q-c channels $\mc{E}_1,\ldots,\mc{E}_t: B(\mc{H}_A)\rightarrow B(\C^V)$ such that
$\rho_1,\ldots,\rho_t$ are perfectly distinguishable, where
    \begin{equation*}
         \rho_i= \left(\Phi_G\otimes \id_H  \right)\circ \left( \mc{E}_i\otimes \id_H \right)(\ket{\psi}\bra{\psi}),
    \end{equation*}
and $\mc{E}_i(\rho)=\sum_{v\in V}\Tr(P^i_v\rho)\ket{v}\bra{v}$ for all $\rho\in B(\mc{H}_A)$, where $\{P^i_v:\,v\in V\}$ is a PVM for all $i\in [t]$.

For every $i\neq j$,
\begin{align*}
    0 &= \Tr(\rho_i\rho_j)\\
    & = \sum_{v,w\in V}\Tr\big(  P^i_v    P^j_w   \big)\abs{\braket{\tilde{\psi}_v}{\tilde{\psi}_w}}^2 \\
    & = \sum_{v\in V} \Tr\big(  P^i_v  \  P^j_v   \big) + \sum_{\substack{v\neq w\\ (v,w)\in E}}\Tr\big(  P^i_v   P^j_w  \big).
\end{align*}
 Here, the first equality uses that $\rho_i$ and $\rho_j$ are perfectly distinguishable, and the last equality uses that $\braket{\tilde{\psi}_v}{\tilde{\psi}_w}=0$ if $v\neq w$ and $(v,w)\notin E$. It follows that for all $i\neq j$,
\begin{equation}
     P^i_v   P^j_w  =0 \label{eq:0}
\end{equation}
if $v=w$ or $(v,w)\in E$. This implies $\alpha_q(G)\geq t$.

We conclude that $C^{(1)}_{0,\PME}(\Phi_G)= \log \left(\alpha_q(G)\right)$.
\end{proof}

\bibliographystyle{bibtex/bst/alphaarxiv.bst}
\bibliography{bibtex/bib/full.bib,bibtex/bib/quantum.bib,bibtex/qref}

\end{document}